\documentclass[a4paper,UKenglish]{lipics-v2019}

\newcommand{\A}{\mathbf{A}}
\newcommand{\N}{{\mathbb{N}}}
\newcommand{\Z}{{\mathbb{Z}}}
\newcommand{\Q}{{\mathbb{Q}}}
\newcommand{\R}{{\mathbb{R}}}
\newcommand{\C}{{\mathbb{C}}}
\newcommand{\ew}{{\varepsilon}}

\newcommand{\rk}[1]{\textsf{Rk}(#1)}
\renewcommand{\phi}{\varphi}
\newcommand{\SL}{\mathrm{SL}(2,\Z)}

\newcommand{\0}{\mathbf{O}}
\newcommand{\I}{\mathbf{I}}
\newcommand{\F}{\mathcal{F}}

\newtheorem{prob}[theorem]{Problem}

\title{On the Mortality Problem: \newline from multiplicative matrix equations to linear recurrence sequences and beyond}

\titlerunning{On the Mortality Problem}

\author{Paul C. Bell}{Department of Computer Science, Byrom Street, Liverpool~John~Moores~University, Liverpool, L3-3AF, UK}{p.c.bell@ljmu.ac.uk}{ https://orcid.org/0000-0003-2620-635X}{}

\author{Igor Potapov}{Department of Computer Science, Ashton Building, Ashton Street, University of
Liverpool, Liverpool, L69-3BX,
UK}{potapov@liverpool.ac.uk}{https://orcid.org/0000-0002-7192-7853}{Partially supported by EPSRC grants EP/R018472/1 and EP/M00077X/1.}

\author{Pavel Semukhin}{Department of Computer Science, University of Oxford, Wolfson Building,
Parks Road, Oxford, OX1 3QD,
UK}{pavel.semukhin@cs.ox.ac.uk}{https://orcid.org/0000-0002-7547-6391}{Supported by ERC grant AVS-ISS (648701).}

\authorrunning{P.\,C. Bell, I. Potapov and P. Semukhin}

\Copyright{Paul C. Bell, Igor Potapov and Pavel Semukhin}

\ccsdesc{Theory of computation~Computability}
\ccsdesc{Mathematics of computing~Computations on matrices}

\keywords{Linear recurrence sequences, Skolem's problem, mortality problem, matrix equations,
primary decomposition theorem, Baker's theorem}

\acknowledgements{We thank Prof. James Worrell for useful discussions, particularly related to
S-unit equations.}

\nolinenumbers
\hideLIPIcs

\begin{document}

\maketitle
\begin{abstract}
We consider the following variant of the Mortality Problem: given $k\times k$
matrices $A_1, A_2, \dots,A_{t}$, does there exist nonnegative integers
$m_1, m_2, \dots,m_t$ such that the product $A_1^{m_1} A_2^{m_2} \cdots A_{t}^{m_{t}}$
is equal to the zero matrix? It is known that this problem is decidable when
$t \leq 2$ for matrices over algebraic numbers but becomes undecidable for
sufficiently large $t$ and $k$ even for integral matrices.

In this paper, we prove the first decidability results for $t>2$. We show as one of our central
results that for $t=3$ this problem in any dimension is Turing equivalent to the well-known Skolem
problem for linear recurrence sequences. Our proof relies on the Primary Decomposition Theorem for
matrices that was not used to show decidability results in matrix semigroups before. As a corollary
we obtain that the above problem is decidable for $t=3$ and $k \leq 3$ for matrices over algebraic
numbers and for $t=3$ and $k=4$ for matrices over real algebraic numbers. Another consequence is
that the set of triples $(m_1,m_2,m_3)$ for which the equation $A_1^{m_1} A_2^{m_2} A_3^{m_3}$
equals the zero matrix is equal to a finite union of direct products of semilinear sets.

For $t=4$ we show that the solution set can be non-semilinear, and thus it seems unlikely that there is a direct connection to the Skolem problem. However we prove that the problem
is still decidable for upper-triangular $2 \times 2$ rational matrices by
employing powerful tools from transcendence theory such as Baker's theorem and
S-unit equations.
\end{abstract}

\section{Introduction}
A large number of naturally defined matrix problems  are still unanswered,
despite the long history of matrix theory.  Some of these questions have recently
drawn renewed interest in the context of the analysis of digital processes, verification problems, and 
links with several  fundamental questions in mathematics \cite{CFKL,BJK05,OW12,OW14,OW14b,OSW15,GO15,COW16,COW16b,OP16,BHP17, PS_SODA2017, ICALPKNP18}.

One of these challenging problems is the {\sl Mortality Problem} of whether the zero matrix belongs to a finitely generated matrix semigroup. It plays a central role in many questions from 
control theory and software verification \cite{Ve85,BT99,BBK01,OP16,BKKW16}. The mortality problem has been known to be undecidable for matrices in $\Z^{3 \times 3}$ since 1970 \cite{Paterson} and the current undecidability bounds for the $M(d, k \times k)$ problem (i.e.\ the mortality problem for semigroups generated by $d$ matrices of size $k\times k$) are $M(6, 3 \times 3)$, $M(4, 5 \times 5)$, $M(3, 9 \times 9)$ and $M(2, 15 \times 15)$, see \cite{CHHN14}.
It is also known that the problem is NP-hard
for $2 \times 2$ integer matrices \cite{BHP12} and is decidable for $2 \times 2$ integer matrices with determinant $0,\pm1$ \cite{NR08}. 
In the case of finite matrix semigroups of any dimension
the mortality problem is known to be PSPACE-complete~\cite{KRS09}.

In this paper, we study a very natural variant of the mortality problem when matrices must appear in a fixed order (i.e. under bounded language constraint):
{\sl Given $k\times k$ matrices $A_1, A_2, \dots,A_{t}$ over a ring $\F$, do there exist $m_1, m_2, \dots,m_{t}\in \N$ such that
$A_1^{m_1} A_2^{m_2} \dots A_{t}^{m_{t}}=\0_{k, k}$, where $\0_{k, k}$ is $k\times k$ zero matrix?
}

In general (i.e. replacing $\0_{k, k}$ by other matrices) this problem is known as the solvability of multiplicative matrix equations 
and has been studied for many decades. In its simplest form, when $k = 1$, the problem was studied by
Harrison in 1969 \cite{Harrison1969} as a reformulation of the ``accessibility problem'' for linear sequential machines. The case $t = 1$ was solved in polynomial time in a celebrated paper by Kannan and Lipton
in 1980 \cite{Kannan1980}. The case $t = 2$,  i.e. $A^{x}{B}^{y}=C$ where $A$, $B$ and $C$ are commuting matrices
was solved by Cai, Lipton and Zalcstein in 1994 \cite{CaiLiptonZalcstein94}. Later, in 1996, the solvability of matrix equations
over commuting matrices was solved in polynomial time in \cite{BBC94} and in 2010 it was shown
in \cite{BeHaHaKaPo} that $A^{x}{B}^{y}=C$ is decidable for non-commuting matrices 
of any dimension with algebraic coefficients by a reduction to the  
commutative case from \cite{BBC94}. However, it was also shown in \cite{BeHaHaKaPo}  that the solvability of multiplicative matrix equations for sufficiently large natural numbers $t$ and $k$ is in general undecidable 
by an encoding of \emph{Hilbert's tenth problem}
 and in particular for the mortality problem with bounded language constraint.
In 2015 it was also shown that the undecidability result holds for such equations with unitriangluar matrices \cite{unitriangular2015}
and also in the case of specific equations with nonnegative matrices \cite{HONKALA2015}.

The decidability of matrix equations for non-commuting matrices is only known as corollaries of either 
recent decidability results for solving membership problem in $2 \times 2$ matrix semigroups  \cite{PS_SODA2017,PS_MFCS2017} or in the case of quite
restricted classes of matrices, e.g. matrices from the Heisenberg group \cite{ICALPKNP18,KLZ2016} or row-monomial matrices over commutative semigroups \cite{LisitsaPotapov2004}.
In the other direction, progress has been made for matrix-exponential equations, but again in the case of commuting matrices \cite{OP16}.

In this paper, we prove the first decidability results for the above problem
when $t=3$ and $t=4$. We will call these problems the ABC-Z
and ABCD-Z problems, respectively. More
precisely, we will show that the ABC-Z problem in any dimension is Turing equivalent to the Skolem
problem (also known as Skolem-Pisot problem) which asks whether a given linear recurrence sequence
ever reaches zero. Our proof relies on the Primary Decomposition Theorem for
matrices (Theorem \ref{prdecthm}) that was not used to show decidability results in matrix semigroups
before.
As a corollary, we obtain that the ABC-Z problem is decidable for $2 \times 2$ and $3 \times 3$ matrices over algebraic numbers and also for $4 \times 4$ matrices over real algebraic numbers.
Another consequence of the above equivalence is that 
the set of triples $(m,n,\ell)$ that satisfy the equation
$A^mB^nC^\ell=\0_{k\times k}$ can be expressed as a finite union of direct
products of semilinear sets.

In contrast to the ABC-Z problem, we show that the solution set of the ABCD-Z
problem can be non-semilinear. This indicates that the ABCD-Z problem is unlikely
to be related to the Skolem problem. However we will show that the ABCD-Z problem
is decidable for upper-triangular $2 \times 2$ rational matrices.  The proof of
this result relies on
powerful tools from transcendence theory such as Baker's theorem for linear forms in logarithms, 
S-unit equations from algebraic number theory and the Frobenius rank inequality
from matrix analysis. More precisely, we will reduce the ABCD-Z equation for
upper-triangular $2 \times 2$ rational matrices to an equation of the form
$ax+by+cz=0$, where $x,y,z$ are S-units, and then use an upper bound on the
solutions of this equation (as in Theorem \ref{thm:sunit}). On the other hand, if we try to
generalize this result to arbitrary $2\times 2$ rational matrices or to
upper-triangular matrices of higher dimension, then we end up with an equation that
contain a sum of four or more S-units, and for such equations no effective upper
bounds on their solutions are known. So, these generalizations seems to lie
beyond the reach of current mathematical knowledge.

\section{Preliminaries}

We denote by $\N$, $\Z$, $\Q$ and $\C$ the sets of natural, integer, rational and complex numbers, respectively. Further, we denote by $\A$ the set of algebraic numbers and by $\A_\R$ the set of real algebraic numbers.

For a prime number $p$ we define a valuation $v_p(x)$ for nonzero $x\in \Q$ as follows: if $x=p^k\frac{m}{n}$, where $m,n\in \Z$ and $p$ does not divide $m$ or $n$, then $v_p(x) = v_p(p^k\frac{m}{n})=k$.

Throughout this paper $\F$ will denote either the ring of integers $\Z$ or one
of the fields $\Q$, $\A$, $\A_\R$ or $\C$. We will use the notation $\F^{n
\times m}$ for the set of $n \times m$ matrices over $\F$.

We denote by $\mathbf{e}_i$ the $i$'th standard basis vector of some dimension
(which will be clear from the context). Let $\0_{n,m}$ be the zero matrix of
size $n\times m$, $\I_n$ be the identity matrix of size $n\times n$, and
$\mathbf{0}_n$ be the zero column vector of length $n$.
Given a finite set of matrices $\mathcal{G} \subseteq \F^{n \times n}$, we
denote by $\langle\mathcal{G}\rangle$ the multiplicative semigroup generated by $\mathcal{G}$. 

If $A \in\F^{m\times m}$ and $B\in\F^{n\times n}$, then we define their direct sum as $A\oplus B=
\left[\begin{array}{@{}c|l@{}}
A & \0_{m,n}\\
\hline
\0_{n,m} & B
\end{array}\right]$.
Let $C \in \F^{k \times k}$ be a square matrix. We write
$\det(C)$ for the
determinant of $C$. We call $C$ singular if $\det(C)= 0$, otherwise it is
said to be invertible (or non-singular). Matrices $A$ and $B$ from $\F^{k \times
k}$ are called \emph{similar} if there exists an invertible $k\times k$ matrix
$S$ (perhaps over a larger field containing $\F$) such that $A = SBS^{-1}$. In
this case, $S$ is said to be a similarity matrix transforming $A$ to $B$.

We will also require the following inequality regarding ranks of matrices, known
as the \emph{Frobenius rank inequality}, see \cite{HJ90} for further details.

\begin{theorem}[Frobenius Rank Inequality]\label{frithm}
Let $A, B, C \in \F^{k \times k}$. Then
\[
\rk{AB} + \rk{BC} \leq \rk{ABC} + \rk{B}
\]
\end{theorem}

In the proof of our first main result about the ABC-Z problem we will make use of the primary decomposition theorem for matrices.

\begin{theorem}[Primary Decomposition Theorem \cite{HoffKun}]\label{prdecthm}
  Let $A$ be a matrix from $\F^{n\times n}$, where $\F$ is a field. Let $m_A(x)$ be the minimal polynomial for $A$ such that
  \[
    m_A(x)=p_1(x)^{r_1}\cdots p_k(x)^{r_k},
  \]
  where the $p_i(x)$ are distinct irreducible monic polynomials over $\F$ and the $r_i$ are positive integers. Let $W_i$ be the null space of $p_i(A)^{r_i}$ and let $S_i$ be a basis for $W_i$. Then
  \begin{enumerate}[(1)]
    \item $S_1\cup \cdots \cup S_k$ is a basis for $\F^n$ and $\F^n=W_1\oplus \cdots \oplus W_k$,
    \item each $W_i$ is invariant under $A$, that is, $A\mathbf{x}\in W_i$ for any $\mathbf{x}\in W_i$,
    \item let $S$ be a matrix whose columns are equal to the basis vectors from $S_1\cup \cdots \cup S_k$; then
      \[
	S^{-1}AS=A_1\oplus \cdots \oplus A_k,
      \]
      where each $A_i$ is a matrix over $\F$ of the size $|S_i|\times |S_i|$, and the minimal polynomial of $A_i$ is equal to $p_i(x)^{r_i}$.
  \end{enumerate}
\end{theorem}

We will also need the following two propositions.

\begin{proposition}\label{prop:dec}
  If $p(x)$ is a polynomial over a field $\F$, where $\F$ is either $\Q$, $\A$ or $\A_\R$, then the primary decomposition of $p(x)$ can be algorithmically computed.
\end{proposition}

\begin{proof}
  If $\F=\Q$, then one can use an LLL algorithm \cite{LLL} to find primary decomposition in polynomial time.

  If $\F=\A$, then one can use well-known algorithms to compute standard representations of the roots of $p(x)$ in polynomial time \cite{BPR06,Co93,Pan96,Pink76}. Let $\lambda_1,\dots,\lambda_k$ be distinct roots of $p(x)$ with multiplicities $m_1,\dots,m_k$, respectively. In this case the primary decomposition of $p(x)$ is equal to
  \[
    p(x) = (x-\lambda_1)^{m_1}\cdots(x-\lambda_k)^{m_k}.
  \]

  If $\F=\A_\R$, then again one can compute in polynomial time standard representations of the roots of $p(x)$ in $\A$. Let $\lambda_1,\dots,\lambda_i$ be real roots of $p(x)$ with multiplicities $m_1,\dots,m_i$ and let $\mu_1,\overline{\mu}_1,\dots,\mu_j,\overline{\mu}_j$ be pairs of complex conjugate roots of $p(x)$ with multiplicities $n_1,\dots,n_j$, respectively. Then the primary decomposition of $p(x)$ over $\A_\R$ is equal to
  \[
    p(x)=(x-\lambda_1)^{m_1}\cdots(x-\lambda_i)^{m_i}p_1(x)^{n_1}\cdots p_j(x)^{n_j},
  \]
  where $p_s(x)=(x-\mu_s)(x-\overline{\mu}_s) = x^2-2\operatorname{Re}(\mu_s)x + {|\mu_s|}^2$ for $s=1,\dots,j$.
\end{proof}

\begin{proposition}\label{prop:min}
  Let $A\in \F^{n\times n}$ and $m_A(x)$ be the minimal polynomial of $A$. Then
  $A$ is invertible if and only if $m_A(x)$ has nonzero free coefficient, i.e.,
  $m_A(x)$ is not divisible by $x$.
\end{proposition}

\begin{proof}
  Suppose that $A$ is invertible but $m_A(x)=xm'(x)$ for some polynomial $m'(x)$. Then
  \[
    \0_{n,n}=m_A(A)=A\cdot m'(A).
  \]
  Multiplying the above equation by $A^{-1}$ we obtain $m'(A)=\0_{n,n}$, which contradicts the assumption that $m_A(x)$ is the minimal polynomial for $A$.

  On the other hand, it $x$ does not divide $m_A(x)$, then $m_A(x)=xm'(x)+a$ for some $m'(x)$ and a nonzero constant $a$. Then
  \[
    \0_{n,n} = m_A(A) = A\cdot m'(A) + a\I_n.
  \]
  From this equation we conclude that $A$ is invertible, and $A^{-1}=-\frac{1}{a}m'(A)$.
\end{proof}

Our proof of the decidability of ABCD-Z problem for $2\times 2$ upper-triangular rational matrices
relies on the following result which is proved using Baker's theorem on linear forms in logarithms
(see Corollary 4 in \cite{EGST} and also \cite{G08}).

\begin{theorem}  \label{thm:sunit}
  Let $S=\{p_1,\dots,p_s\}$ be a finite collection of prime numbers and let $a,b,c$ be relatively
  prime nonzero integers, that is, $\gcd(a,b,c)=1$.
  
  If $x,y,z$ are relatively prime nonzero integers composed of primes from
  $S$ that satisfy the equation $ax+by+cz=0$, then
  \[
    \max\{|x|,|y|,|z|\} < \exp(s^{Cs} P^{4/3} \log A)
  \]
  for some constant $C$, where $P = \max\{p_1,\dots,p_s\}$ and $A=\max\{|a|,|b|,|c|,3\}$.
\end{theorem}

\begin{remark*}
Rational numbers whose numerator and denominator are divisible only by the primes
from $S$ are called \emph{S-units}.
\end{remark*}

\section{Linear recurrence sequences and semilinear sets}

There is a long history in computer science and mathematics of studying sequences of numbers defined by some recurrence relation, where the next value in the sequence depends upon some `finite memory' of previous values in the sequence. Possibly the simplest, and certainly the most well known of these, is the \emph{Fibonacci sequence}, which may be defined by the recurrence $F(n) = F(n-1) + F(n-2)$ with $F(0) = F(1) = 1$ being given as the \emph{initial conditions} of the sequence. 
We may generalise the Fibonacci sequence to define a \emph{linear recurrence sequence}, which find application in many areas of mathematics and other sciences and for which many questions remain open. Let $\F$ be a ring; a sequence $(u_n)_{n=0}^{\infty}$ is called a \emph{linear recurrence sequence} ($1$-LRS) if it satisfies a relation of the form:
\[
u_n = a_{k-1}u_{n-1} + \cdots +   a_{1}u_{n-k+1} +  a_{0}u_{n-k},
\]
for any $n \geq k$, where each $a_0, a_1, \ldots, a_{k-1} \in \F$ are fixed coefficients\footnote{In the literature, such a sequence is ordinarily called an LRS; we use the nomenclature $1$-LRS since we will study a multidimensional variant of this concept. Also, 1-LRS are usually considered over integers, but in the present paper we will consider such sequences over algebraic numbers.}. Such a sequence $(u_n)_{n=0}^{\infty}$ is said to be of depth $k$ if it satisfies no shorter linear recurrence relation (for any $k' < k$). We call the initial $k$ values of the sequence $u_0, u_1, \ldots, u_{k-1}$ the initial conditions of the $1$-LRS. Given the initial conditions and coefficients of a 1-LRS, every element is uniquely determined.

The \emph{zero set} of a 1-LRS is defined as follows: $\mathcal{Z}(u_n) = \{j
\in \N\ |\ u_j = 0\}$.

There are various questions that one may ask regarding $\mathcal{Z}(u_n)$. One notable example relates to the famous ``Skolem's problem'' which is stated in the following way:

\begin{prob}[Skolem's Problem]
Given the coefficients and initial conditions of a depth $k$ 1-LRS $(u_n)_{n=0}^{\infty}$, determine if $\mathcal{Z}(u_n)$ is the empty set.
\end{prob}

Skolem's problem has a long and rich history, see \cite{HHHK05} for a good survey. We note here that the problem remains open despite properties of zero sets having been studied even since 1934 \cite{Sk34}. It is known that the Skolem problem is at least NP-hard \cite{BP02} and that it is decidable for depth $3$ over $\A$ and for depth $4$ over $\A_\R$, see \cite{Ve85} and \cite{MST84}\footnote{A proof of decidability for depth $5$ was claimed in \cite{HHHK05}, although there is possibly a gap in the proof \cite{OW12}.}. Other interesting questions are related to the structure of $\mathcal{Z}(u_n)$. We remind the reader the definition of semilinear sets.

\begin{definition}[Semilinear set]
  A set $S \subseteq \N$ is called \emph{semilinear} if it is the union of a finite set and finitely many arithmetic progressions.
\end{definition}

A seminal result regarding $1$-LRSs is that there zero sets are semilinear.

\begin{theorem}[Skolem, Mahler, Lech \cite{Ma35, Sk34, Le53} and \cite{HHHK05, Hans}]\label{smlthm}
  The zero set of a 1-LRS over $\C$ (or more generally over any field of characteristic 0) is semilinear.

In particular, if $(u_n)_{n=0}^{\infty}$ is a 1-LRS whose coefficients and initial conditions are algebraic numbers, then one can algorithmically find a number $L\in \N$ such that for every $i=0,\dots,L-1$, if we let $u^i_m=u_{i+mL}$, then
\begin{enumerate}[(1)]
  \item the sequence $(u^i_m)_{m=0}^{\infty}$ is a 1-LRS of the same depth as $(u_n)_{n=0}^{\infty}$, and
  \item either $\mathcal{Z}(u^i_m)=\N$ or $\mathcal{Z}(u^i_m)$ is finite.
\end{enumerate}
\end{theorem}

Note that in the above theorem we can decide whether $\mathcal{Z}(u^i_m)$ is finite or $\mathcal{Z}(u^i_m)=\N$ because $\mathcal{Z}(u^i_m)=\N$ if and only if $u^i_0=\cdots =u^i_{k-1}=0$, where $k$ is the depth of $(u^i_m)_{m=0}^{\infty}$.

We will also consider a stronger version of the Skolem problem.

\begin{prob}[Strong Skolem's Problem]
  Given the coefficients and initial conditions of a 1-LRS $(u_n)_{n=0}^{\infty}$ over $\A$, find a description of the set $\mathcal{Z}(u_n)$. That is, find a finite set $F$ such that $\mathcal{Z}(u_n) = F$ if $\mathcal{Z}(u_n)$ is finite or, if $\mathcal{Z}(u_n)$ is infinite, find a finite set $F$, a constant $L\in \N$ and numbers $i_1,\dots,i_t\in \{0,\dots,L-1\}$ such that
  \[
    \mathcal{Z}(u_n) = F \cup \{i_1+mL : m\in \N\} \cup \cdots \cup \{i_t+mL : m\in \N\}.
  \]
\end{prob}

Using the Skolem-Mahler-Lech theorem we can prove an equivalence between the strong version of the
Skolem problem and the standard version\footnote{This result was announced in \cite{Ve85} without a
proof, probably with a similar construction in mind.}.

\begin{theorem}\label{strsk}
  Skolem's problem of depth $k$ over $\A$ is Turing equivalent to the strong Skolem's problem of the same depth.
\end{theorem}

\begin{proof}
  Obviously, Skolem's problem is reducible to the strong Skolem's problem. We now show a reduction in the other direction.

  Let $(u_n)_{n=0}^{\infty}$ be a depth-$k$ 1-LRS over $\A$. By Theorem \ref{smlthm}, we can algorithmically find a number $L$ such that, for every $i=0,\dots,L-1$, the sequence $u^i_m=u_{i+mL}$ is a 1-LRS of depth~$k$ which is either everywhere zero, that is, $\mathcal{Z}(u^i_m)=\N$ or $\mathcal{Z}(u^i_m)$ is finite. Recall that we can decide whether $\mathcal{Z}(u^i_m)$ is equal to $\N$ by considering the first $k$ terms of $(u^i_m)_{m=0}^{\infty}$.

  By definition, we have $\mathcal{Z}(u_n) = \bigcup\limits_{i=0}^{L-1}
  \{i+L\!\cdot\!\mathcal{Z}(u^i_m)\}$. So, if $\mathcal{Z}(u^i_m)=\N$, then $\{i+L\!\cdot\!\mathcal{Z}(u^i_m)\} = \{i+mL : m\in \N\}$, and if $\mathcal{Z}(u^i_m)$ is finite, then so is $\{i+L\!\cdot\!\mathcal{Z}(u^i_m)\}$.

  To finish the proof we need to show how to compute $\mathcal{Z}(u^i_m)$, and hence $\{i+L\!\cdot\!\mathcal{Z}(u^i_m)\}$, when it is finite. For this we will use an oracle for the Skolem problem. Let $m'$ be the smallest number such that $\mathcal{Z}(u^i_{m+m'})$ is empty. Such $m'$ exists because $\mathcal{Z}(u^i_m)$ is finite. Furthermore, $(u_{m+m'})_{m=0}^{\infty}$ is a 1-LRS of depth $k$ for any $m'$. So, we ask the oracle for the Skolem problem to decide whether $\mathcal{Z}(u^i_{m+m'})=\emptyset$ for each $m'\in \N$ starting from $0$ until we find one for which $\mathcal{Z}(u^i_{m+m'})$ is empty. Note that we do not have any bound on $m'$ because we do not even know the size of $\mathcal{Z}(u^i_m)$. All we know is that $\mathcal{Z}(u^i_m)$ is finite, and hence the above algorithm will eventually terminate.
  Since $\mathcal{Z}(u^i_m)$ is a subset of $\{0,\dots,m'\}$, then we can compute it by checking whether $u^i_m=0$ for $m=0,\dots,m'$.
\end{proof}

Linear recurrence sequences can also be represented using matrices \cite{HHHK05}:

\begin{lemma}\label{lrsequiv}
  Let $\F$ be a ring; for a sequence $(u_n)_{n=0}^{\infty}$ over $\F$ the following are equivalent:
\begin{enumerate}[(1)]
\item $(u_n)_{n=0}^{\infty}$ is a 1-LRS of depth $k$.
\item There are vectors $\mathbf{u},\mathbf{v} \in \F^k$ and a matrix $M\in \F^{k\times k}$ such that $u_n = \mathbf{u}^TM^n\mathbf{v}$ for $n\in \N$.
\end{enumerate}

Moreover, for any matrix $M\in \F^{k\times k}$, the sequence $u_n=
\left(M^n\right)_{[1,k]}$ is a 1-LRS of depth at most $k$. On the other hand, if $(u_n)_{n=0}^{\infty}$ is a 1-LRS of depth $k$, then there is a matrix $M\in \F^{(k+1)\times (k+1)}$ such that $u_n = \left(M^n\right)_{[1,k+1]}$ for all $n\in \N$.
\end{lemma}

Lemma~\ref{lrsequiv} motivates the following definition of $n$-dimensional Linear Recurrence Sequences ($n$-LRSs) which as we show later are related to the mortality problem for bounded languages.

\begin{definition}[$n$-LRS]\label{nlrsdef}
A multidimensional sequence $u_{m_1, m_2, \ldots, m_n}$ is called an $n$-LRS of depth $k$ over $\F$ if there exist two vectors $\mathbf{u},\mathbf{v} \in \F^k$ and matrices $M_1, M_2, \ldots, M_n \in \F^{k \times k}$ such that 
\[
u_{m_1, m_2, \ldots, m_n} = \mathbf{u}^T M_1^{m_1} M_2^{m_2} \cdots M_n^{m_n} \mathbf{v}.
\]
\end{definition}

\section{The mortality problem for bounded languages}

We remind the reader the definition of the mortality problem for bounded
languages.

\begin{prob}[Mortality for bounded languages]\label{mortbounded}
Given $k\times k$ matrices $A_1,\dots,A_{t}$ over a ring $\F$, do there exist $m_1, m_2, \dots,m_{t}\in \N$ such that
\[
A_1^{m_1} A_2^{m_2} \dots A_{t}^{m_{t}}=\0_{k, k}.
\]
\end{prob}

Recall that for $t = 3$ and $t=4$ this problem is called the ABC-Z and ABCD-Z problem,
respectively. 
Our first main result is that the ABC-Z problem is 
computationally equivalent to the
Skolem problem for $1$-LRS. Our reduction holds in any dimension 
and over the same number field which means
that any new decidability results for the Skolem problem will automatically extend the decidability 
of ABC-Z equations and can immediately  lead to new decidability results for equations in dimensions 2, 3 and 4. 
For the proof we will need the following technical lemma.

\begin{lemma} \label{lem:abc}
  Let $\F$ be a field, and suppose $A, B, C \in \F^{k \times k}$ are matrices of the form
  \[
    A=
    \left[\begin{array}{@{}c|c@{}}
	A_{s,s} & \0_{s,k-s}\\
	\hline
	\0_{k-s,s} & \0_{k-s,k-s}
    \end{array}\right],\
    B=
    \left[\begin{array}{@{}c|c@{}}
	B_{s,t} & X_{s,k-t}\\
	\hline
	Y_{k-s,t} & Z_{k-s,k-t}
    \end{array}\right],\
    C=
    \left[\begin{array}{@{}c|c@{}}
	C_{t,t} & \0_{t,k-t}\\
	\hline
	\0_{k-t,t} & \0_{k-t,k-t}
    \end{array}\right]
  \]
  for some $s,t\leq k$, where $A_{s,s}$, $B_{s,t}$, $X_{s,k-t}$, $Y_{k-s,t}$, $Z_{k-s,k-t}$ and $C_{t,t}$ are matrices over $\F$ whose dimensions are indicated by their subscripts (in particular, $A=A_{s,s}\oplus \0_{k-s,k-s}$ and $C=C_{t,t}\oplus \0_{k-t,k-t}$). If $A_{s,s}$ and $C_{t,t}$ are invertible matrices, then the equation $ABC=\0_{k,k}$ is equivalent to $B_{s,t} = \0_{s,t}$.
\end{lemma}

\begin{proof}
  It is not hard to check that
  \[
    AB=
    \left[\begin{array}{@{}c|c@{}}
	A_{s,s} & \0_{s,k-s}\\
	\hline
	\0_{k-s,s} & \0_{k-s,k-s}
    \end{array}\right]\cdot
    \left[\begin{array}{@{}c|c@{}}
	B_{s,t} & X_{s,k-t}\\
	\hline
	Y_{k-s,t} & Z_{k-s,k-t}
    \end{array}\right]
    =
    \left[\begin{array}{@{}c|c@{}}
	A_{s,s}B_{s,t} & A_{s,s}X_{s,k-t}\\
	\hline
	\0_{k-s,t} & \0_{k-s,k-t}
    \end{array}\right],
  \]
  and hence
  \[
    (AB)C=
    \left[\begin{array}{@{}c|c@{}}
	A_{s,s}B_{s,t} & A_{s,s}X_{s,k-t}\\
	\hline
	\0_{k-s,t} & \0_{k-s,k-t}
    \end{array}\right]\cdot
    \left[\begin{array}{@{}c|c@{}}
	C_{t,t} & \0_{t,k-t}\\
	\hline
	\0_{k-t,t} & \0_{k-t,k-t}
    \end{array}\right]
    =
    \left[\begin{array}{@{}c|c@{}}
	A_{s,s}B_{s,t}C_{t,t} & \0_{s,k-t}\\
	\hline
	\0_{k-s,t} & \0_{k-s,k-t}
    \end{array}\right].
  \]
  So, if $B_{s,t}=\0_{s,t}$, then $ABC=\0_{k,k}$. Conversely, if $ABC=\0_{k,k}$, then $A_{s,s}B_{s,t}C_{t,t}=\0_{s,t}$. Using the fact that $A_{s,s}$ and $C_{t,t}$ are invertible matrices, we can multiply the equation $A_{s,s}B_{s,t}C_{t,t}=\0_{s,t}$ by $A^{-1}_{s,s}$ on the left and by $C^{-1}_{t,t}$ on the right to obtain that $B_{s,t}=\0_{s,t}$.
\end{proof}

The next lemma is similar to Lemma \ref{lem:abc} and can also be proved by directly multiplying
the matrices.

\begin{lemma} \label{lem:ab}
  (1) Suppose $A, B\in \F^{k \times k}$ are matrices of the following form
  \[
    A=
    \left[\begin{array}{@{}c|c@{}}
	A_{s,s} & \0_{s,k-s}\\
	\hline
	\0_{k-s,s} & \0_{k-s,k-s}
    \end{array}\right] = A_{s,s}\oplus \0_{k-s,k-s}\quad \text{and}\quad
    B=
    \left[\begin{array}{@{}c}
	B_{s,k}\\
	\hline
	X_{k-s,k}
    \end{array}\right]\!,\
  \]
  for some $s\leq k$. If $A_{s,s}$ is invertible, then $AB=\0_{k,k}$ is equivalent to $B_{s,k} = \0_{s,k}$.

  (2) Suppose $A, B\in \F^{k \times k}$ are matrices of the following form
  \[
    A=\left[ \begin{array}{@{}c|c@{}} A_{k,t} & Y_{k,k-t}\end{array}\right]\quad \text{and}\quad
    B=
    \left[\begin{array}{@{}c|c@{}}
	B_{t,t} & \0_{t,k-t}\\
	\hline
	\0_{k-t,t} & \0_{k-t,k-t}
    \end{array}\right] = B_{t,t}\oplus \0_{k-t,k-t},
  \]
  for some $t\leq k$. If $B_{t,t}$ is invertible, then $AB=\0_{k,k}$ is equivalent to $A_{k,t} = \0_{k,t}$.

  As in Lemma \ref{lem:abc}, in the above equations $A_{s,s}$, $B_{s,k}$, $X_{k-s,k}$, $A_{k,t}$, $Y_{k,k-t}$ and $B_{t,t}$ are matrices over $\F$ whose dimensions are indicated by their subscripts.
\end{lemma}

\begin{theorem}\label{thm:ABC}
  Let $\F$ be the ring of integers $\Z$ or one of the fields\/ $\Q$, $\A$ or $\A_\R$. Then the ABC-Z
  problem for matrices from $\F^{k\times k}$ is Turing equivalent to the Skolem problem of depth $k$ over $\F$.
\end{theorem}

\begin{proof}
  First, we show reduction from the ABC-Z problem to the Skolem problem.

  Clearly, the ABC-Z problem over $\Z$ is equivalent to the ABC-Z problem over $\Q$ (by multiplying the
  matrices $A,B,C$ by a suitable integer number). It is also not hard to see that the Skolem problem
  for 1-LRS over $\Q$ is equivalent to the Skolem problem over $\Z$ for 1-LRS of the same depth.
  Indeed, by Lemma \ref{lrsequiv} we can express any 1-LRS $(u_n)_{n=0}^\infty$ over $\Q$ as $u_n =
  \mathbf{u}^TM^n\mathbf{v}$ for some rational vectors $\mathbf{u}$ and $\mathbf{v}$ and a rational
  matrix $M$. If we multiply $\mathbf{u}$, $\mathbf{v}$ and $M$ by a suitable natural number $t$,
  then $(t^{n+2}u_n)_{n=0}^\infty$ will be an integer 1-LRS, which has the same zero set as
  $(u_n)_{n=0}^\infty$. Hence, without loss of generality, we will assume that $\F$ is one of the
  fields $\Q$, $\A$ or $\A_\R$.

  Consider an instance of the ABC-Z problem: $A^mB^nC^\ell=\0_{k,k}$, where $A,B,C\in \F^{k,k}$. Let $\chi_A(x)$ be the characteristic polynomial of $A$. By Proposition \ref{prop:dec}, we can find a primary decomposition $\chi_A(x)=p_1(x)^{m_1}\cdots p_t(x)^{m_t}$, where $p_1(x),\dots,p_t(x)$ are distinct irreducible monic polynomials. From this decomposition we can find the minimal polynomial $m_A(x)$ of $A$ because $m_A(x)$ is a factor of $\chi_A(x)$, and we can check all divisors of $\chi_A(x)$ to find $m_A(x)$.

  Let $m_A(x)=p_1(x)^{r_1}\cdots p_u(x)^{r_u}$, where $p_1(x),\dots,p_u(x)$ are distinct irreducible
  monic polynomials. Now we apply the Primary Decomposition Theorem (Theorem \ref{prdecthm}) to $A$.
  Let $S_i$ be a basis for the null space of $p_i(A)^{r_i}$, which can be found, e.g., using
  Gaussian elimination. Let $S$ be a matrix whose columns are the vectors of the basis
  $S_1\cup\cdots\cup S_u$. Then
  \[
    S^{-1}AS = A_1\oplus \cdots \oplus A_u,
  \]
  where the minimal polynomial of $A_i$ is $p_i(A)^{r_i}$ for $i=1,\dots,u$. Similarly, we can compute a primary decomposition $m_C(x)=q_1(x)^{s_1}\cdots q_v(x)^{s_v}$ of the minimal polynomial for $C$, where $q_1(x),\dots,q_v(x)$ are distinct irreducible monic polynomials, and a matrix $T$ such that
  \[
    T^{-1}CT = C_1\oplus \cdots \oplus C_v,
  \]
  where the minimal polynomial of $C_i$ is $q_i(C)^{s_i}$ for $i=1,\dots,v$.

  Note that if $p(x)$ is an irreducible monic polynomial, then either $p(x)=x$ or $x$ does not divide $p(x)$. So, among the polynomials $p_1(x),\dots,p_u(x)$ in the primary decomposition of $m_A(x)$ at most one is equal to $x$, and the same holds for the polynomials $q_1(x),\dots,q_v(x)$ in the primary decomposition of $m_C(x)$.

  Suppose, for example, that $p_u(x)=x$. In this case $m_A(x)=p_1(x)^{r_1}\cdots p_{u-1}(x)^{r_{u-1}}x^{r_u}$, and $S^{-1}AS = A_1\oplus \cdots \oplus A_{u-1}\oplus A_u$, where the minimal polynomial of $A_u$ is $x^{r_u}$, and hence $A_u$ is a nilpotent matrix of index $r_u$. Recall that, for $i=1,\dots,u-1$, the polynomial $p_i(x)$ is not divisible by $x$, and so is $p_i(x)^{r_i}$, which is the minimal polynomial for $A_i$. Hence, by Proposition~\ref{prop:min}, $A_i$ is invertible. Let $A_\mathrm{inv} = A_1\oplus \cdots \oplus A_{u-1}$ and $A_\mathrm{nil} = A_u$. Then we obtain
  \begin{equation} \label{eq:A}
    S^{-1}AS = A_\mathrm{inv}\oplus A_\mathrm{nil},
  \end{equation}
  where $A_\mathrm{inv}$ is invertible, and $A_\mathrm{nil}$ is nilpotent. If
$p_i(x)=x$ for some $i<u$, then we need in addition to permute some rows and columns of matrix $S$ to obtain one that gives us Equation~(\ref{eq:A}) above. If none of the $p_i(x)$ is equal to $x$, then we assume that $A_\mathrm{nil}$ is the empty matrix of size $0\times 0$.

  The same reasoning can be applied to matrix $C$, that is,
we can compute an invertible matrix $C_\mathrm{inv}$, a nilpotent (or empty) matrix $C_\mathrm{nil}$, and an invertible matrix $T$ such that
 \[
    T^{-1}CT = C_\mathrm{inv}\oplus C_\mathrm{nil}.
 \]
  Note that the indices of the nilpotent matrices $A_\mathrm{nil}$ and $C_\mathrm{nil}$ are at most $k$, and hence $A^k_\mathrm{nil}$ and $C^k_\mathrm{nil}$ are zero (or empty) matrices.

  Our goal is to find all triples $(m,n,\ell)\in \N^3$ for which $A^mB^nC^\ell=\0_{k,k}$. In order to do this we will consider four cases: (1) $m\geq k$ and $\ell\geq k$, (2) $m<k$ and $\ell<k$, (3) $m\geq k$ and $\ell<k$, and (4) $m<k$ and $\ell\geq k$.

  Before dealing with each of these cases, we note that the equation $A^mB^nC^\ell=\0_{k,k}$ is equivalent to
  \begin{align*}
    S(A^m_\mathrm{inv}\oplus A^m_\mathrm{nil})S^{-1}&B^nT(C^\ell_\mathrm{inv}\oplus C^\ell_\mathrm{nil})T^{-1} =\0_{k,k}\quad \text{or to}\\
    (A^m_\mathrm{inv}\oplus A^m_\mathrm{nil})S^{-1}&B^nT(C^\ell_\mathrm{inv}\oplus C^\ell_\mathrm{nil}) =\0_{k,k}
  \end{align*}
  because $S$ and $T$ are invertible matrices.
  
  Now suppose $A_\mathrm{inv}$ has size $s\times s$, and $C_\mathrm{inv}$ has size $t\times t$ for some $s,t\leq k$.
  
  \textbf{Case 1: $m\geq k$ and $\ell\geq k$.} Since $m,\ell\geq k$, we have $A^m_\mathrm{nil} = \0_{k-s,k-s}$ and $C^\ell_\mathrm{nil} = \0_{k-t,k-t}$, and hence the equation $A^mB^nC^\ell=\0_{k,k}$ is equivalent to
  \begin{equation} \label{eq:B}
    (A^m_\mathrm{inv}\oplus \0_{k-s,k-s})S^{-1}B^nT(C^\ell_\mathrm{inv}\oplus \0_{k-t,k-t}) =\0_{k,k}.
  \end{equation}
  Suppose the matrix $S^{-1}B^nT$ has a form $S^{-1}B^nT=
    \left[\begin{array}{@{}c|c@{}}
	B_{s,t} & X_{s,k-t}\\
	\hline
	Y_{k-s,t} & Z_{k-s,k-t}
    \end{array}\right]$. Since $A^m_\mathrm{inv}$ and $C^\ell_\mathrm{inv}$ are invertible matrices, Lemma \ref{lem:abc} implies that Equation (\ref{eq:B}) is equivalent to $B_{s,t}=\0_{s,t}$. Therefore, we obtain the following equivalence: $A^mB^nC^\ell=\0_{k,k}$ if and only if
    \begin{equation}\label{eq:s}
      s^{i,j}_n = (\mathbf{e}_i^\top S^{-1})B^n(T\mathbf{e}_j) = 0\quad \text{for all}\ i=1,\dots,s\ \text{and}\ j=1,\dots,t.
    \end{equation}
  By Lemma \ref{lrsequiv}, the sequence $(s^{i,j}_n)_{n=0}^\infty$ is a 1-LRS of order $k$ over $\F$. As in the proof of Theorem~\ref{strsk}, we can use an oracle for the Skolem problem for 1-LRS of depth $k$ over $\F$ to compute the descriptions of the semilinear sets $\mathcal{Z}(s^{i,j}_n)$. Hence we can compute a description of the intersection $Z_1=\bigcap\limits_{\substack{i=1,\dots,s\\ j=1,\dots,t}}\mathcal{Z}(s^{i,j}_n)$, which is also a semilinear set. An important observation is that the set $Z_1$ does not depend on $m$ and $\ell$.

  \textbf{Case 2: $m<k$ and $\ell<k$.} Fix some $m<k$ and $\ell<k$. For this particular choice of $m$ and $\ell$, the equation $A^mB^nC^\ell=\0_{k,k}$ is equivalent to
  \[
    s^{i,j}_n = (\mathbf{e}_i^\top A^m)B^n(C^\ell\mathbf{e}_j) = 0\quad \text{for all}\ i=1,\dots,k\ \text{and}\ j=1,\dots,k.
  \]
  Again, by Lemma \ref{lrsequiv}, the sequence $(s^{i,j}_n)_{n=0}^\infty$ is a 1-LRS of order $k$ over $\F$, and we can use an oracle for the Skolem problem for 1-LRS of depth $k$ over $\F$ to compute the descriptions of the semilinear sets $\mathcal{Z}(s^{i,j}_n)$. Therefore, we can compute a description of the intersection $Z_2(m,\ell)=\bigcap\limits_{\substack{i=1,\dots,k\\ j=1,\dots,k}}\mathcal{Z}(s^{i,j}_n)$ which is equal to all values of $n$ for which $A^mB^nC^\ell=\0_{k,k}$ holds for fixed $m,\ell<k$.

  \textbf{Case 3: $m\geq k$ and $\ell<k$.} To solve this case we will combine ideas from cases (1) and~(2). Fix some $\ell<k$ and let $m$ be any integer such that $m\geq k$. Then $A^m_\mathrm{nil} = \0_{k-s,k-s}$, and the equation $A^mB^nC^\ell=\0_{k,k}$ is equivalent to
  \begin{equation} \label{eq:C}
    (A^m_\mathrm{inv}\oplus \0_{k-s,k-s})S^{-1}B^nC^\ell =\0_{k,k}.
  \end{equation}
  Suppose the matrix $S^{-1}B^nC^\ell$ has a form $S^{-1}B^nC^\ell=
    \left[\begin{array}{@{}c}
	B_{s,k}\\
	\hline
	X_{k-s,k}
    \end{array}\right]$. Since $A^m_\mathrm{inv}$ is invertible, Lemma \ref{lem:ab} implies that Equation (\ref{eq:C}) is equivalent to $B_{s,k}=\0_{s,k}$. Therefore, Equation (\ref{eq:C}) is equivalent to
  \[
    s^{i,j}_n = (\mathbf{e}_i^\top S^{-1})B^n(C^\ell\mathbf{e}_j) = 0\quad \text{for all}\ i=1,\dots,s\ \text{and}\ j=1,\dots,k.
  \]
  As in the previous two cases, we can use an oracle for the Skolem problem for 1-LRS of depth $k$ over $\F$ to compute the descriptions of the semilinear sets $\mathcal{Z}(s^{i,j}_n)$ and of the intersection $Z_3(\ell)=\bigcap\limits_{\substack{i=1,\dots,s\\ j=1,\dots,k}}\mathcal{Z}(s^{i,j}_n)$. $Z_3(\ell)$ is the set of all $n$ for which $A^mB^nC^\ell=\0_{k,k}$ holds for a fixed $\ell<k$ and for any $m\geq k$.

  \textbf{Case 4: $m<k$ and $\ell\geq k$.} Fix some $m<k$ and let $\ell$ be any integer such that $\ell\geq k$. Using the same ideas as in Case 3 we can use an oracle for the Skolem problem to compute a description of a semilinear set $Z_4(m)$ which is equal to all values of $n$ for which $A^mB^nC^\ell=\0_{k,k}$ holds for a fixed $m<k$ and for any $\ell\geq k$.

  Combining all the above cases together, we conclude that the set of all triples $(m,n,\ell)\in \N^3$ that satisfy the equation $A^mB^nC^\ell=\0_{k,k}$ is equal to the following union
  \begin{equation}\label{eq:semilin}
    \begin{aligned}
      \{(m,n,\ell) : n\in Z_1 \text { and } m,\ell\geq k\} &\bigcup \bigcup_{m,\ell<k} \{(m,n,\ell) : n\in Z_2(m,\ell)\} \bigcup\\ \bigcup_{\ell<k} \{(m,n,\ell) : n\in Z_3(\ell) \text { and } m\geq k\} &\bigcup \bigcup_{m<k} \{(m,n,\ell) : n\in Z_4(m) \text { and } \ell\geq k\}.
    \end{aligned}
  \end{equation}

  Having a description for the above set, we can decide whether is it empty or not, that is, whether there exist $m,n,\ell\in \N$ such that $A^mB^nC^\ell=\0_{k,k}$.

  \medskip
  We now show the reduction in the other direction. Let $(u_n)_{n=0}^\infty$ be a 1-LRS that satisfies a relation
\[
u_n = a_{k-1}u_{n-1} + \cdots +   a_{1}u_{n-k+1} +  a_{0}u_{n-k},
\]
where $a_0\neq 0$. Let $A$, $B$ and $C$ be the following matrices of size $k\times k$:
\[
A = \begin{pmatrix} u_{k-1} & \cdots & u_1 & u_0 \\ 0 & \cdots & 0 & 0 \\ 0 & \cdots & 0 & 0 \\  \vdots & \ddots & \vdots & \vdots \\ 0 & \cdots & 0 & 0 \end{pmatrix}\!,\
B = \begin{pmatrix} a_{k-1} & 1 & \cdots & 0 & 0 \\ \vdots & \vdots & \ddots & \vdots & \vdots \\ a_2 & 0 & \cdots & 1 & 0 \\  a_1 & 0 & \cdots & 0 & 1 \\ a_0 & 0 & \cdots & 0 & 0 \end{pmatrix}\!,\
C = \begin{pmatrix} 0 & 0 & \cdots & 0 & 0 \\ 0 & 0 & \cdots & 0 & 0 \\ \vdots & \vdots & \ddots & \vdots & \vdots \\  0 & 0 & \cdots & 0 & 0 \\ 0 & 0 & \cdots & 0 & 1 \end{pmatrix}\!.
\]
A straightforward computation shows that the product $A^mB^nC^\ell$ is equal to a matrix where all entries equal zero except for the entry in the upper-right corner which is equal to $u_{k-1}^{m-1}u_n$. So, if we assume that $u_{k-1}\neq 0$, then we have the following implications: (1) if $A^mB^nC^\ell = \0_{k, k}$ for some $m,n,\ell\in \N$ with $m,\ell\geq 1$, then $u_n=0$; and (2) if $u_n=0$, then the equation $A^mB^nC^\ell = \0_{k, k}$ holds for any $m,\ell\geq 1$.

The assumption that $u_{k-1}\neq 0$ is not a serious restriction because we can shift the original
sequence by at most $k$ positions to ensure that $u_{k-1}\neq 0$. In other words, instead of $(u_n)_{n=0}^\infty$ we can consider a sequence $(u_{n+t})_{n=0}^\infty$ for some $t>0$. It is easy to check that a 1-LRS of depth $k$ is identically zero if and only if it contains $k$ consecutive zeros. Hence if $(u_n)_{n=0}^\infty$ is not identically zero, then we can find $t<k$ such that the sequence $(u_{n+t})_{n=0}^\infty$ satisfies the condition that $u_{k-1+t}\neq 0$.
\end{proof}

\begin{corollary}\label{cor:semil}
  The set of triples $(m,n,\ell)$ that satisfy an equation $A^mB^nC^\ell=\0_{k, k}$ is equal to a finite union of direct products of semilinear sets.
\end{corollary}

\begin{proof}
  The corollary follows from Equation (\ref{eq:semilin}) above that describes all triples $(m,n,\ell)$ that satisfy the equation $A^mB^nC^\ell=\0_{k, k}$. By construction and the Skolem-Mahler-Lech theorem, the sets $Z_1$, $Z_2(m,\ell)$, $Z_3(\ell)$ and $Z_4(m)$ are semilinear. In Equation (\ref{eq:semilin}) we take direct product of these sets either with singleton sets or with sets of the form $\N_k = \{n\in \N : n\geq k\}$, which are also semilinear sets, and then take a finite union of such products.
  In other words, Equation (\ref{eq:semilin}) can be rewritten as follows
  \begin{align*}
    &\N_k \times Z_1 \times \N_k\ \bigcup \bigcup_{m,\ell<k} \{m\}\times
    Z_2(m,\ell)\times \{\ell\}\ \bigcup\\ &\bigcup_{\ell<k} \N_k\times
    Z_3(\ell)\times \{\ell\}\ \bigcup \bigcup_{m<k} \{m\} \times Z_4(m)\times
    \N_k.
  \end{align*}
The main corollary of Theorem \ref{thm:ABC} is the following result.
\end{proof}

\begin{corollary}
The ABC-Z problem is decidable for $3\times 3$ matrices over algebraic numbers and for $4\times 4$ matrices over real algebraic numbers.
\end{corollary}

\begin{proof}
  By Theorem \ref{thm:ABC}, the ABC-Z problem for $3\times 3$ matrices over $\A$ is equivalent to the Skolem problem of depth 3 over $\A$, and the ABC-Z problem $4\times 4$ matrices over $\A_\R$ is equivalent to the Skolem problem of depth 4 over $\A_\R$.
  Now the corollary follows from the fact that the Skolem problem is decidable for linear recurrence sequences of depth 3 over $\A$ and of depth 4 over $\A_\R$ \cite{Ve85,MST84}.
\end{proof}

The next theorem is a generalisation of Theorem \ref{thm:ABC} to an arbitrary number of matrices.

\begin{theorem}\label{thm:nLRS}
  Let $\F$ be the ring of integers $\Z$ or one of the fields\/ $\Q$, $\A$ or $\A_\R$. Then the mortality problem for bounded languages (Problem~\ref{mortbounded}) over $\F$ for $t+2$ matrices is reducible to the following problem: given matrices $A_1,\dots,A_{t}$ from $\F^{k\times k}$ and vectors $\mathbf{u}_i$, $\mathbf{v}_i$ from $\F^k$, where $i=1,\dots,r$ ,do there exists $m_1,\dots,m_{t}\in \N$ such that $s^i_{m_1,\dots,m_t}=0$ for all $i=1,\dots,r$, where
\[
  s^i_{m_1,\dots,m_t}=\mathbf{u}_i^\top A_1^{m_1} \dots A_{t}^{m_{t}}\mathbf{v}_i.
\]
\end{theorem}

\begin{proof}
  Consider an instance  of the mortality problem for bounded languages for $t+2$ matrices:
  \[
    A_0^{m_0}A_1^{m_1}\cdots A_t^{m_t}A_{t+1}^{m_{t+1}}=\0_{k, k}.
  \]
  The proof of the reduction is similar to the proof of Theorem \ref{thm:ABC} for the equation $A^mB^nC^\ell=\0_{k, k}$. However note that in Theorem \ref{thm:ABC} we proved a stronger result in the sense that in Equation (\ref{eq:semilin}) we gave a description of all the triples $(m,n,\ell)$ that satisfy $A^mB^nC^\ell=\0_{k, k}$. If we simply want to know whether there \emph{exists} one such triple, then we only need to consider Case 1 from the proof of Theorem \ref{thm:ABC}, because if the equation $A^mB^nC^\ell=\0_{k, k}$ has a solution, then it has one in which $m\geq k$ and $\ell\geq k$.
  
  So, we replicate the proof of Case 1 where in place of matrices $A$ and $C$ we consider $A_0$ and $A_{t+1}$. By doing this we obtain a system of equations similar to Equation (\ref{eq:s}), where in place of $B^n$ we will have the product $A_1^{m_1}\cdots A_t^{m_t}$. This gives us the desired reduction.
  
  The key difference between this theorem and Theorem \ref{thm:ABC} is that for $t=1$ the system of equations (\ref{eq:s}) for 1-LRSs $s^i_n$ can be reduced to the Skolem problem using Theorem \ref{strsk} and the Skolem-Mahler-Lech theorem (Theorem \ref{smlthm}). For $t>1$, the solution set of the equation
  \[
    \mathbf{u}^\top A_1^{m_1} \dots A_{t}^{m_{t}}\mathbf{v}=0
  \]
  is not semilinear (see, e.g., Proposition \ref{nonlrsset}), and we do not have an analog of Theorem \ref{strsk} or the Skolem-Mahler-Lech theorem in this case. So, we cannot solve a system of such equations using an oracle for a single equation of this form.
\end{proof}

\section{The ABCD-Z problem in dimension two}
Recall that the ABCD-Z problem in dimension two asks whether there exist natural numbers $k,m,n,\ell\in \N$ such that
\begin{equation}
  A^k B^m C^n D^\ell = \0_{2,2}. \label{abcdeq}
\end{equation}
In this section we will show that this problem is decidable for $2 \times 2$ upper-triangular matrices with rational coefficients.
In the proof we will use the following simple lemmas which show how to diagonalise and compute powers of an upper-triangular $2 \times 2$ matrix.

\begin{lemma}\label{simillem}
  Given an upper triangular matrix $\begin{pmatrix} a & b \\ 0 & c \end{pmatrix}$ such that $a\neq c$ then:\\
$
\begin{pmatrix} a & b \\ 0 & c \end{pmatrix} = \begin{pmatrix} 1 & -\frac{b}{a-c} \\ 0 & 1 \end{pmatrix} \begin{pmatrix} a & 0 \\ 0 & c \end{pmatrix} \begin{pmatrix} 1 & \frac{b}{a-c} \\ 0 & 1 \end{pmatrix}
$
\end{lemma}

\begin{lemma}\label{geolem}
  For any matrices of the form $\begin{pmatrix} a & b \\ 0 & c \end{pmatrix}$ and
  $\begin{pmatrix} a & b \\ 0 & a \end{pmatrix}$, such that $a\neq c$, and any $k\in \N$ we
  have
$
{\begin{pmatrix} a & b \\ 0 & c \end{pmatrix}}^k = \begin{pmatrix} a^k & b\dfrac{a^k-c^k}{a-c} \\ 0 & c^k \end{pmatrix} \text{ and } 
{\begin{pmatrix} a & b \\ 0 & a \end{pmatrix}}^k = \begin{pmatrix} a^k& kba^{k-1} \\ 0 & a^k \end{pmatrix}.
$
\end{lemma}

\begin{theorem}\label{abcdthm}
The ABCD-Z problem is decidable for $2 \times 2$ upper-triangular matrices over rational numbers.
\end{theorem}

\begin{proof}
  First, note that if one of the matrices $A$, $B$, $C$ or $D$ is nilpotent, then Equation~(\ref{abcdeq}) obviously has a solution. So from now on we assume that none of $A$, $B$, $C$ or $D$ is nilpotent. Furthermore, if $A$ or $D$ is invertible, then the ABCD-Z problem reduces to the ABC-Z problem for rational matrices of dimension two that is decidable by Theorem \ref{thm:ABC}. So, without loss of generality, we will assume that both $A$ and $D$ are singular matrices.
  
  Now suppose we are given an instance of the ABCD-Z problem which satisfies the above-mentioned requirements. We will show that if Equation~(\ref{abcdeq}) has a solution, then it has a solution with $k=\ell=1$. Indeed, assume Equation~(\ref{abcdeq}) has a solution, and let $(k,m,n,\ell)$ be a solution of minimal length, where the length of a solution is the sum $k+m+n+\ell$. Using Theorem \ref{thm:ABC} we can exclude the case when $k=0$ or $\ell=0$ since in this case our problem is just an instance of the ABC-Z for $2\times 2$ matrices. So we will assume that in the above solution $k,\ell \geq 1$.

By the Frobenius rank inequality (Theorem~\ref{frithm}), we have that:
\[
  \rk{A^{k}B^{m}C^{n}D^{\ell-1}} + \rk{A^{k-1}B^{m}C^{n}D^{\ell}} \leq \rk{A^{k-1}B^{m}C^{n}D^{\ell-1}}.
\]
This follows since $\rk{A^{k}B^{m}C^{n}D^{\ell}} = \rk{\0_{2,2}} = 0$. In the above inequality, notice that neither $A^{k}B^{m}C^{n}D^{\ell-1}$ nor $A^{k-1}B^{m}C^{n}D^{\ell}$ is the zero matrix by the assumption that the solution has minimal length. Hence the ranks of the matrices on the left hand side are at least~$1$. Therefore, $\rk{A^{k-1}B^{m}C^{n}D^{\ell-1}} = 2$. Since we assumed that $A$ and $D$ are singular, it is necessary that $k = \ell = 1$. Also notice that if $\rk{B}=1$ or $\rk{C}=1$, then we must have $m=0$ or $n=0$, respectively. Again these cases can be excluded by Theorem~\ref{thm:ABC}.

Thus, using the Frobenius rank inequality and the assumption that the solution is of minimal length, we reduced the ABCD-Z problem to an equation of the form:
\[
  AB^{m}C^{n}D = \0_{2,2},
\]
where $\rk{A} = \rk{D} = 1$ and $\rk{B} = \rk{C} = 2$.

We assumed that $A$ and $D$ have rank one but are not nilpotent. This means that they have one zero and one nonzero element on the diagonal, in particular, they satisfy the condition of Lemma~\ref{simillem}. Hence we can find invertible rational matrices $S_A$ and $S_D$ such that
\[
  A = S_A^{-1} \begin{pmatrix} a & 0 \\ 0 & 0 \end{pmatrix}S_A\quad \text{and}\quad
  D = S_D^{-1} \begin{pmatrix} d & 0 \\ 0 & 0 \end{pmatrix}S_D,
\]
where $a,d$ are nonzero rational numbers. Applying Lemma \ref{lem:abc} with
matrix $S_A B^m C^n S_D^{-1}$ in place of $B$, we conclude that $AB^{m}C^{n}D =
\0_{2,2}$ holds if and only if the $(1,1)$-entry of $S_A B^m C^n S_D^{-1}$ equals zero.
In other words, the equation $AB^{m}C^{n}D = \0_{2,2}$ is equivalent to $s_{m,n} = \mathbf{u}^\top B^m C^n \mathbf{v} = 0$, where $\mathbf{u}^\top = \mathbf{e}_1^\top S_A$ and $\mathbf{v} = S_D^{-1}\mathbf{e}_1$ are vectors with rational coefficients.

We will consider three cases: (1) both $B$ and $C$ have distinct eigenvalues, (2) both $B$ and $C$ have a single eigenvalue of multiplicity 2 and (3) one matrix has distinct eigenvalues and the other has a single eigenvalue of multiplicity 2.

\textbf{Case (1):} $B$ and $C$ have distinct eigenvalues, that is, $B=\begin{pmatrix} b_1 & b_3 \\ 0 & b_2 \end{pmatrix}$ and $C=\begin{pmatrix} c_1 & c_3 \\ 0 & c_2 \end{pmatrix}$, where $b_1\neq b_2$ and $c_1\neq c_2$. By Lemma \ref{geolem} we have
\[
  B^m = \begin{pmatrix} b_1^m & b_3\dfrac{b_1^m-b_2^m}{b_1-b_2} \\ 0 & b_2^m \end{pmatrix}\quad \text{and}\quad 
  C^n = \begin{pmatrix} c_1^n & c_3\dfrac{c_1^n-c_2^n}{c_1-c_2} \\ 0 & c_2^n \end{pmatrix}
\]
Multiplying these matrices we obtain:  $B^mC^n=$
\noindent
\begin{align*}
 = \begin{pmatrix} b_1^m & b_3\dfrac{b_1^m-b_2^m}{b_1-b_2} \\ 0 & b_2^m \end{pmatrix} \begin{pmatrix} c_1^n & c_3\dfrac{c_1^n-c_2^n}{c_1-c_2} \\ 0 & c_2^n \end{pmatrix}
 = \begin{pmatrix} b_1^mc_1^n & b_3\dfrac{b_1^mc_2^n-b_2^mc_2^n}{b_1-b_2} + c_3\dfrac{b_1^mc_1^n-b_1^mc_2^n}{c_1-c_2} \\ 0 & b_2^mc_2^n \end{pmatrix}
\end{align*}
From this formula one can see that the entries of $B^mC^n$ are linear combinations of $b_1^mc_1^n$, $b_1^mc_2^n$ and $b_2^mc_2^n$ with rational coefficients. Notice that the term $b_2^mc_1^n$ does not appear in the entries of $B^mC^n$. Since $\mathbf{u}$ and $\mathbf{v}$ are rational vectors we conclude that
\[
  s_{m,n} = \mathbf{u}^\top B^m C^n \mathbf{v} = \alpha b_1^mc_1^n +\beta b_1^mc_2^n +\gamma b_2^mc_2^n,
\]
where $\alpha,\beta,\gamma\in \Q$. Multiplying the equation $s_{m,n} = \alpha b_1^mc_1^n +\beta b_1^mc_2^n +\gamma b_2^mc_2^n=0$ by the product of denominators of $\alpha,\beta,\gamma$ and $b_1,b_2,c_1,c_2$ we can rewrite it in the following form
\[
  as^mr^n+bs^mt^n+cq^mt^n=0
\]
where $a,b,c,q,r,s,t$ are integers and $a,b,c$ are relatively prime. Recall that we want to find out
if there exist $m,n\in \N$ that satisfy the above exponential Diophantine equation. If one of the
coefficients $a,b,c$ is zero, then this problem is easy to solve. For instance, if $b=0$ then the
above equation is equivalent to $as^mr^n = -cq^mt^n$. This equality holds if and only if $as^mr^n$
and $-cq^mt^n$ have the same sign and $v_p(as^mr^n) = v_p(-cq^mt^n)$ for every prime divisor $p$ of
$a,c,s,r,q$ or $t$. Each of these conditions can be expressed as a linear Diophantine equation. For
instance, the sign of $as^mr^n$ or $-cq^mt^n$ depends on the parity of $m$ and $n$. So, the
requirement that $as^mr^n$ and $-cq^mt^n$ have the same sign can be written as a linear congruence
equation in $m$ and $n$ modulo $2$, which in turn can be expressed as a linear Diophantine equation.
Since a system of linear Diophantine equations can be effectively solved, we can find out whether
there exist $m$ and $n$ that satisfy $as^mr^n = -cq^mt^n$.

Now suppose that $a,b,c$ are relatively prime nonzero integers. Let
$T$ be all primes that appear in $s,r,q$ or $t$. Theorem
\ref{thm:sunit} gives an upper bound on nonzero relatively prime integers
$x,y,z$ that are composed of the primes from $T$ and satisfy the equation
$ax+by+cz=0$. Therefore, we can algorithmically compute the following set
\[
  \mathcal{U}=\{(x,y,z)\ :\ ax+by+cz=0,\ x,y,z\neq 0 \text{ and } \gcd(x,y,z)=1\}.
\]
Next for each triple $(x,y,z)\in \mathcal{U}$ we want to find out if there exist $m,n\in \N$ such that
\begin{equation} \label{eq:dioph}
  (s^mr^n,s^mt^n,q^mt^n)= (xg,yg,zg)
\end{equation}
for some $g\in \N$ that is composed of the primes from $T$. It is not hard to see that Equation~(\ref{eq:dioph}) holds if and only if for every $p\in T$
\[
  v_p(s^mr^n)-v_p(x) = v_p(s^mt^n)-v_p(y) = v_p(q^mt^n)-v_p(z)
\]
and $s^mr^n,s^mt^n,q^mt^n$ have the same signs as $x,y,z$, respectively. Since these conditions can be expressed as a system of linear Diophantine equations, we can algorithmically find if there are $m,n\in \N$ that satisfy Equation~(\ref{eq:dioph}). If such $m$ and $n$ exist for at least one triple $(x,y,z)\in \mathcal{U}$, then the original equation $s_{m,n} = 0$ has a solution. Otherwise, the equation $s_{m,n} = 0$ does not have a solution.

\textbf{Case (2):} both $B$ and $C$ have a single eigenvalue of multiplicity 2, that is, $B=\begin{pmatrix} b_1 & b_2 \\ 0 & b_1 \end{pmatrix}$ and $C=\begin{pmatrix} c_1 & c_2 \\ 0 & c_1 \end{pmatrix}$. By Lemma \ref{geolem} we have
\[
  B^m = \begin{pmatrix} b_1^m & mb_2b_1^{m-1} \\ 0 & b_1^m \end{pmatrix}\quad \text{and}\quad 
  C^n = \begin{pmatrix} c_1^n & nc_2c_1^{n-1} \\ 0 & c_1^n \end{pmatrix}.
\]
Multiplying these matrices we obtain:
\[
  B^mC^n = \begin{pmatrix} b_1^m & mb_2b_1^{m-1} \\ 0 & b_1^m \end{pmatrix} \begin{pmatrix} c_1^n & nc_2c_1^{n-1} \\ 0 & c_1^n \end{pmatrix}
  = \begin{pmatrix} b_1^mc_1^n & mb_2b_1^{m-1}c_1^n + nc_2b_1^mc_1^{n-1} \\ 0 & b_1^mc_1^n \end{pmatrix}.
\]
Note that the entries of $B^mC^n$ are equal to linear combinations of $b_1^mc_1^n$, $mb_1^{m-1}c_1^n$ and $nb_1^mc_1^{n-1}$ with rational coefficients. Therefore,
\[
  s_{m,n} = \mathbf{u}^\top B^m C^n \mathbf{v} = \alpha b_1^mc_1^n +\beta mb_1^{m-1}c_1^n +\gamma nb_1^mc_1^{n-1} = b_1^{m-1}c_1^{n-1}(\alpha b_1c_1 +m\beta c_1 +n\gamma b_1)
\]
where $\alpha,\beta,\gamma\in \Q$. Hence $s_{m,n}=0$ is equivalent to the linear Diophantine equation
\[
  \alpha b_1c_1 +m\beta c_1 +n\gamma b_1=0,
\]
which can be easily solved.

\textbf{Case (3):} one matrix has distinct eigenvalues and the other has a single eigenvalue of multiplicity 2. We consider only the case when $B$ has distinct eigenvalues and $C$ has a single eigenvalue because the other case is similar.

Suppose $B=\begin{pmatrix} b_1 & b_3 \\ 0 & b_2 \end{pmatrix}$, where $b_1\neq b_2$, and $C=\begin{pmatrix} c_1 & c_2 \\ 0 & c_1 \end{pmatrix}$. By Lemma \ref{geolem} we have
\[
  B^m = \begin{pmatrix} b_1^m & b_3\dfrac{b_1^m-b_2^m}{b_1-b_2} \\ 0 & b_2^m \end{pmatrix}\quad \text{and}\quad 
  C^n = \begin{pmatrix} c_1^n & nc_2c_1^{n-1} \\ 0 & c_1^n \end{pmatrix}.
\]
Multiplying these matrices we obtain:
\[
  B^mC^n = \begin{pmatrix} b_1^m & b_3\dfrac{b_1^m-b_2^m}{b_1-b_2} \\ 0 & b_2^m \end{pmatrix} \begin{pmatrix} c_1^n & nc_2c_1^{n-1} \\ 0 & c_1^n \end{pmatrix}
  = \begin{pmatrix} b_1^mc_1^n & b_3\dfrac{b_1^mc_1^n-b_2^mc_1^n}{b_1-b_2} + nc_2b_1^mc_1^{n-1} \\ 0 & b_2^mc_1^n \end{pmatrix}.
\]
Note that the entries of $B^mC^n$ are equal to linear combinations of $b_1^mc_1^n$, $b_2^mc_1^n$ and $nb_1^mc_1^{n-1}$ with rational coefficients. Therefore,
\[
  s_{m,n} = \mathbf{u}^\top B^m C^n \mathbf{v} = \alpha b_1^mc_1^n +\beta b_2^mc_1^n +\gamma nb_1^mc_1^{n-1} = b_1^mc_1^{n-1}(\alpha c_1 +\beta c_1\Big(\frac{b_2}{b_1}\Big)^m \!+\gamma n)
\]
where $\alpha,\beta,\gamma\in \Q$. Hence $s_{m,n}=0$ is equivalent to
\[
  \alpha c_1 +\beta c_1\Big(\frac{b_2}{b_1}\Big)^m \!+\gamma n=0.
\]
If $\beta=0$ then the solution is trivial. If $\beta\neq 0$, then after multiplying the above equation by the product of denominators of $\alpha,\beta,\gamma$ and $c_1$ we can rewrite it as
\[
  c\frac{s^m}{t^m} = a+bn,
\]
where $a,b,c,s,t\in \Z$, $t>0$, $\gcd(s,t)=1$ and $\dfrac{s}{t}=\dfrac{b_2}{b_1}$. If $t>1$, then it is not hard to see that the left hand side can be integer only for finitely many values of $m$, which can be effectively found, and for each of these values of $m$ one can check if there exists corresponding $n\in \N$. If $t=1$, then the above equation reduces to
\[
  cs^m\equiv a \pmod b.
\]
This equation can also be algorithmically solved. For instance, one can start computing the values $cs^m \pmod b$ for $m=0,1,2,\dots$. Since there are only $b-1$ different residues modulo $b$, this sequence will be eventually periodic, and so we can decide whether $cs^m \pmod b$ is equal to $a \pmod b$ for some $m$ and in fact find all such $m$.
\end{proof}

\begin{remark*}
It is interesting to note that in Cases (1) and (2) the solutions $(m,n)$ of the equation $s_{m,n}=0$ are described by linear Diophantine equations, and only in Case (3) we have a linear-exponential equation. This agrees with an example from Proposition~\ref{nonlrsset}, in which matrix $A$ has a single eigenvalue $1$ of multiplicity $2$ and matrix $B$ has distinct eigenvalues $1$ and $2$.
\end{remark*}

\begin{remark*}
In the above argument we used Theorem \ref{thm:sunit} to obtain a bound on the
solutions of the equation $as^mr^n+bs^mt^n+cq^mt^n=0$, which is a special type of an S-unit
equation. This leaves open an interesting question of whether any S-unit equation can be encoded
into the ABCD-Z problem.
\end{remark*}

The obvious question is how hard would it be to solve $n$-LRSs, or in general multiplicative matrix
equations, in low dimensions. In fact we can show that the Skolem problem for $n$-LRSs of depth $2$ is 
NP-hard. It is not direct but an easy corollary following the hardness proof of the mortality problem for $2 \times 2$ matrices \cite{BHP12}.

\begin{theorem}[\cite{BHP12}]
The mortality problem for integer matrices of dimension two is NP-hard.
\end{theorem}

\begin{corollary}\label{cor:NP-h}
Determining if the zero set of an $n$-LRS of depth $2$ is empty is NP-hard.
\end{corollary}

\begin{proof}

In the paper \cite{BHP12} it was shown that the mortality problem is NP-hard for matrices from $\SL$ extended by singular $2 \times 2$ matrices.
Actually without changing the class of matrices from the generator it is possible to 
design a matrix equation with $2 \times 2$ matrices from this generator such that the zero matrix is reachable if and only if
the solution of this matrix equation exist. Let us remind the main parts of the construction
proposed in~\cite{BHP12}.

\begin{figure}[h]
\begin{center}
\includegraphics[scale=0.50]{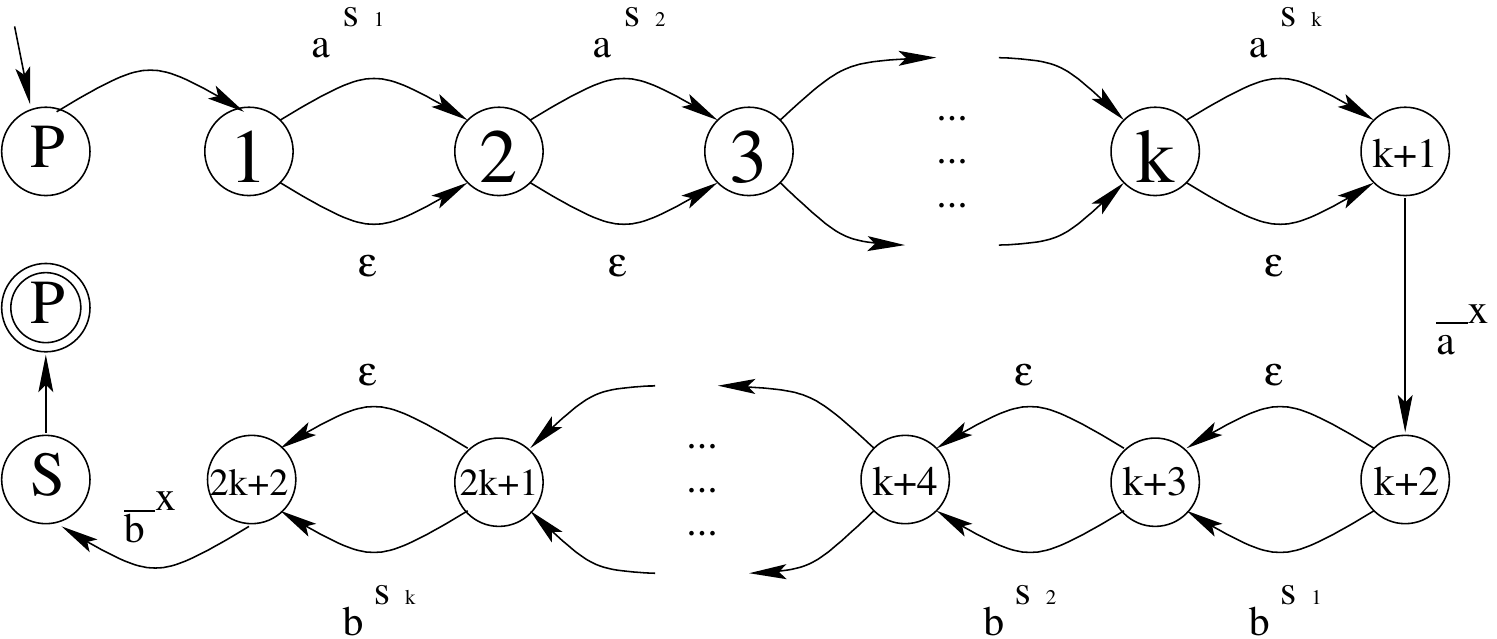}
\caption{The structure of a solution to the mortality problem.}\label{fig_graph}
\end{center}
\end{figure}

The structure of the  zero product can be only reached via a path represented in Figure~\ref{fig_graph} which is encoding of the Subset Sum Problem.
Let us remind that every transition can be represented by a matrix of a polynomial size comparing to binary representation of powers used for  encoding of the Subset Sum Problem with $k$ numbers.  Let us assume that each matrix representing the encoding of the transition from the node $i$ to the node $j$ with a label $c$ is denoted as $M_{i,j,c}$.
Then one can encode in a single equation which is covering all possible paths in the graph represented in Figure~\ref{fig_graph}, where 
alternative pair of paths need to be listed sequentially. Please note that alternative transitions will be used only ones in the equation to follow the path
that can lead to zero matrix. 

Although  we can encode original graph in Figure~\ref{fig_graph} the solution could be modified to make the equation to be shorter
as the second identity does not need to be constructed because the whole length of possible paths can be fixed by the equation length. Here is the form
of the equation with unknowns $y_1, \ldots , y_{2k+3}$ for which the problem is NP-hard: 
\[
  M_{P,1,\ew}^{y_1} M_{1,2,a^{s_1}}^{y_2} M_{1,2,\ew}^{y_3} \cdot \ldots \cdot  M_{k,k+1,a^{s_k}}^{y_{2k}} M_{k,k+1,\ew}^{y_{2k+1}} M_{k+1,S,a^{-x}}^{y_{2k+2}}
M_{S, P,\ew}^{y_{2k+3}} = \0_{2 \times 2}
\]

Since matrix $P$ in the above construction from  \cite{BHP12} is of the form
$\begin{pmatrix} 1 & 0 \\  0 &  0
 \end{pmatrix}$
 the whole product is equal to zero matrix if and only if the element (1,1) is equal to zero.
So determining if the zero set of an $2k+3$-LRS of depth $2$ based on the following form is empty
\[
\begin{pmatrix} 1 & 0 
 \end{pmatrix}
\cdot  M_{P,1,\ew}^{y_1} M_{1,2,a^{s_1}}^{y_2} M_{1,2,\ew}^{y_3} \cdot \ldots \cdot  M_{k,k+1,a^{s_k}}^{y_{2k}} M_{k,k+1,\ew}^{y_{2k+1}} M_{k+1,S,a^{-x}}^{y_{2k+2}}
M_{S, P,\ew}^{y_{2k+3}}  \cdot 
\begin{pmatrix} 1  \\   0
 \end{pmatrix}
\]
should be NP-hard.
\end{proof}

Another interesting observation is that the zero set of a 2-LRS is not necessarily semilinear, in contrast to the situation for 1-LRSs, which indicates that the Skolem problem for 2-LRSs is likely to be significantly harder than the Skolem problem for 1-LRSs even for sequences of small depth.

\begin{proposition}\label{nonlrsset}
There exists a $2$-LRS of depth $2$ for which the zero set is not semilinear. 
\end{proposition}

\begin{proof}

Let $u = (0, 1)^T, v = (1, -1)^T$, $A = \begin{pmatrix} 1 & 0 \\  1 &  1
\end{pmatrix}$ and $B = \begin{pmatrix} 1 & 0 \\  0 &  2 \end{pmatrix}$. Define
\[
  s_{n,m} = u^TA^nB^mv =  (0,1){\begin{pmatrix} 1 & 0 \\  1 & 1
  \end{pmatrix}}^n{\begin{pmatrix} 1 & 0 \\  0 &  2
  \end{pmatrix}}^m\begin{pmatrix} 1 \\ -1 \end{pmatrix} =n-2^m.
\]
Then $s_{n,m}=0$ if and only if $n=2^m$. Clearly, the zero set is not semilinear.
\end{proof}

\bibliography{refs}

\end{document}